%% file: main.tex
\documentclass[aps,prl,twocolumn,superscriptaddress,groupedaddress]{revtex4}  
\usepackage{graphicx}  
\usepackage{dcolumn}   
\usepackage{bm}        
\usepackage{dsfont}
\usepackage{amssymb}   
\usepackage[utf8]{inputenc}
\usepackage[english]{babel}
\usepackage{amsthm}
\usepackage{amsmath}
\usepackage{comment}
\usepackage{cancel}
\usepackage[paperwidth=210mm,paperheight=297mm,centering,hmargin=2cm,vmargin=2.5cm]{geometry}
\usepackage[normalem]{ulem}
\usepackage{color}
\usepackage{braket}
\newtheoremstyle{mytheoremstyle} 
    {\topsep}                    
    {\topsep}                    
    {}                   
    {}                           
    {\itshape}                   
    {.}                          
    {.5em}                       
    {}  

\theoremstyle{mytheoremstyle}
\newtheorem{lemma}{Lemma}

\newtheorem{definition}{Definition}
\newtheorem{corollary}{Corollary}
\newtheorem{theorem}{Theorem}
\newtheorem{postulate}{Postulate}

\usepackage[svgnames]{xcolor}
\usepackage{tikz}
\usetikzlibrary{decorations.markings}
\usetikzlibrary{shapes.geometric}
\usetikzlibrary{circuits.ee.IEC}

\pgfdeclarelayer{edgelayer}
\pgfdeclarelayer{nodelayer}
\pgfsetlayers{edgelayer,nodelayer,main}

\tikzstyle{none}=[]
\tikzstyle{trace}=[circuit ee IEC,thick,ground,rotate=0,scale=2]

\begin{document}
\widetext

\title{Entanglement is an inevitable feature of any non-classical theory}

\author{Jonathan G. Richens\footnotemark[1]}
\affiliation{Controlled Quantum Dynamics theory group, Department of Physics, Imperial College London, London SW7 2AZ, UK.}
\affiliation{Department of Physics and Astronomy, University College London,
Gower Street, London WC1E 6BT, UK.}
\thanks{These authors contributed equally to this work}

\author{John H. Selby\footnotemark[1]}
\affiliation{Controlled Quantum Dynamics theory group, Department of Physics, Imperial College London, London SW7 2AZ, UK.}
\affiliation{University of Oxford, Department of Computer Science, Oxford OX1 3QD, UK.}

\author{Sabri W. Al-Safi}
\affiliation{School of Science \& Technology, Nottingham Trent University, Burton Street, Nottingham NG1 4BU, UK}

\begin{abstract}
One of the most striking features of quantum theory is the existence of entangled states, responsible for Einstein's so called ``spooky action at a distance''. These states emerge from the mathematical formalism of quantum theory, but to date  we do not have a clear idea of which physical principles give rise to entanglement. 
Why does quantum theory have entangled states? Would any theory superseding classical theory have entangled states, or is quantum theory special? We demonstrate that without entanglement, non-classical degrees of freedom cannot reversibly interact. We present two postulates, no-cloning / no-broadcasting and local transitivity, either of which are sufficient to imply the existence of entangled states in any non-classical theory with reversible interactions. Therefore we argue that entanglement is an inevitable feature of a non-classical universe.
\end{abstract}

\maketitle

\section{Introduction}

Entanglement and non-locality are two of the features of quantum theory that clash most strongly with our classical preconceptions as to how the universe works. In particular, they create a tension with the other major theory of the twentieth century: relativity~\cite{einstein1916foundation}. This is most clearly illustrated by Bell's theorem~\cite{bell1966physics,wood2015lesson}, in which certain entangled states are shown to violate local realism. 
In this paper we ask whether entanglement is a surprising feature of nature, or whether it should be expected in any non-classical theory? Could a scientist with no knowledge of quantum theory have predicted the existence of entangled states based solely on the premise that their classical understanding of the world was incomplete? 

To answer these questions we explore how the dynamics of theories are contingent on the existence of entangled states. Specifically, we focus on the existence of reversible interactions, a key feature of both classical and quantum theory. We find that reversible interactions between non-classical degrees of freedom must generate entangled states. From this we then propose a set of physical postulates that imply the existence of entangled states in any probabilistic theory obeying them. Remarkably, aside from the requirement that systems can reversibly interact, these postulates concern only the local properties of the individual systems.

The outline of this paper is as follows. In the following section we discuss the framework, and explain what we mean by entanglement, reversible interactions and classical degrees of freedom within this framework. In the results section we present a novel approach to answering these questions based on the diagrammatic approach to generalised probabilistic theories~\cite{hardy2011reformulating,chiribella2011informational,coecke2017picturing}. Finally in the discussion section we interpret our results and discuss our proposed physical postulates in-depth.

\section{Setup - Generalized probabilistic theories}

In this paper, we will work in the generalised probabilistic theory framework. This framework is broad enough to allow one to discuss essentially arbitrary operationally defined theories, and is based on the idea that any physical theory should be able to predict the outcomes of experiments. We work in the diagrammatic based framework of \cite{hardy2011reformulating,chiribella2011informational} which defines such a theory in terms of a collection of \emph{processes}. These processes correspond to obtaining a particular outcome in a single use of a piece of laboratory apparatus, where the apparatus may have input and output ports for particular types of physical \emph{systems} and a classical pointer to indicate which outcome has occurred.

The above can be formulated via a diagrammatic notation, in which processes are labeled boxes and systems are labeled wires connecting them.  Some examples are shown in the table below, with the quantum-theoretic analogue of each process for comparison.

\begin{center}
\begin{tabular}{c|c|c|c}
& Process & State & Effect \\ \hline
\begin{minipage}{1.4cm}\centering\vspace{-1cm}Diagram\vspace{0.3cm}\end{minipage} & $\input{process.tex}$ & $\input{state.tex}$ & $\input{effect.tex}$ \\ \hline
QT & CP map & \begin{minipage}{1.4cm}\centering\vspace{0.2cm}Density matrix\end{minipage} & \begin{minipage}{1.4cm}\centering\vspace{0.2cm}POVM element\end{minipage}
\end{tabular}
\end{center}

These processes can be connected together acyclically ( i.e., the  output port of one  process may feed into the input port of another  but no loops are allowed) ensuring that systems match, to form experiments. If an experiment has no disconnected ports then we must be able to associate a probability to obtaining any possible set of outcomes, for example:

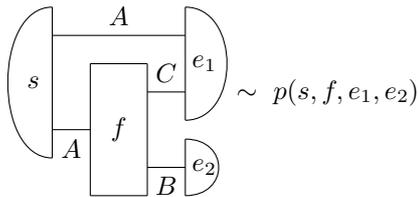
\[\input{experiment.tex}\]

However, if there are disconnected ports then in general this will not be possible, as which outcome occurs could depend on the inputs and outputs of the apparatus.

Given some standard operational assumptions \cite{chiribella2010probabilistic,chiribella2016quantum} we can associate each system to a compact, convex subset of a finite-dimensional real vector space, referred to as the \emph{state space} of the system. For example the Bloch ball is the state space for a qubit, and $d$-dimensional simplices represent the state spaces of $(d+1)$-level classical systems. States then correspond to vectors within this set, and pure states correspond to extremal points of the set. General  transformations are linear maps between the vector spaces containing the convex sets. Effects correspond to positive linear functionals on these sets. There is a unique deterministic effect $u$ that returns $1$ for all vectors in the convex set. Applying this effect to a subsystem of a composite system is analogous to applying the partial trace in quantum theory. This is diagrammatically represented as,
\[\input{trace.tex}\]

The existence of such an effect is often taken as part of the basic framework, but can also be shown to be a consequence of the \emph{causality} axiom used in \cite{chiribella2011informational}.

\begin{definition}\label{def:Entanglement} Entanglement:
A pure state $s$ is entangled when
\[ \input{entanglement.tex} \]
\end{definition}

We can therefore define a composite of two systems such that there are no entangled states. This is known as the minimal tensor product (denoted $\boxtimes$) and is defined as the convex hull of all pure product states.\\

\begin{definition}\label{def:RI} Reversible interactions:
A processes $T$ is reversible if there exists $T^{-1}$ where $T^{-1}\circ T = \mathds{1} = T\circ T^{-1}$. It is a non-trivial interaction if,
\[ \input{interaction1.tex}\]
\end{definition}
Note that some transformations, such as the swap, satisfy this definition of being a reversible interaction, but this could equally well just be a relabeling of our systems and so no genuine correlations are generated between them. As such, in the results section we introduce a refined notion of what it means for a transformation to be a reversible interaction.\\

\begin{definition}\label{def:classicality} Classical degrees of freedom:
a state space $\mathcal{S}$ has a classical degree of freedom iff
\[\exists \, \mathcal{A}\,  \& \, \mathcal{B} \text{ such that } \mathcal{S}= \mathcal{A}\oplus\mathcal{B}\]
where $\mathcal{A}$ and $\mathcal{B}$ are state spaces and $\oplus$ denotes their direct sum (see appendix). Any state space $\mathcal{S}$ with this property is said to be decomposable.

\end{definition}
At first glance this seems like an odd definition of classicality. However, we can interpret $\mathcal{S}=\mathcal{A}\oplus\mathcal{B}$ as a state space where either a state from $\mathcal{A}$ or a state from $\mathcal{B}$ is prepared and there is a classical label as to which of these it belongs to. This classical label is the degree of freedom that is referred to in the above definition. In quantum theory this corresponds to having a state space with super-selection rules \cite{bartlett2007reference} giving a decomposition of the density matrices into block diagonal form.

This definition of classicality becomes particularly clear when we impose the following physical postulate on $\mathcal{S}$ relating the dynamics and kinematics of the theory. 

\begin{postulate}
Transitivity: For any pure states $x,y$ there exists a reversible transformation $T$ s.t. $T(x)=y$
\end{postulate}

We do not assume transitivity in the derivation of our results but employ it later as one of the postulates that is sufficient to derive the existence of entangled states. The usefulness of this Postulate is contained in the following result: if the postulate is satisfied then our above definition of a classical degree of freedom explicitly corresponds to a classical subsystem.\\

\begin{theorem} Transitive state spaces with a classical degree of freedom have a classical subsystem, 
\[\mathcal{S}=\mathcal{A}\oplus \mathcal{B} \implies  \mathcal{S}=\Delta_N\boxtimes \mathcal{C}\] where $\Delta_N$ is some $N+1$-level classical system and $\mathcal{C}$ some state space.
\end{theorem}\label{decomposable transitive}
\proof
See Appendix
\endproof

However, even in cases where we do not assume transitivity it is clear that there is still a classical degree of freedom in the state space, even if it can not be written as an independent subsystem. For example, it has been shown that decomposable state spaces allow for non-disturbing measurements, giving rise to classical observables \cite{barrett2007information}. As such all systems with a classical degree of freedom violate a generalized no-cloning / no-broadcasting postulate on some observables \cite{barnum2007generalized}. Therefore for this work we also propose the following physical postulate to rule out state spaces with a classical degree of freedom. 

\begin{postulate}
No cloneable information: A state space has no cloneable information if it does not permit a non-trivial non-disturbing measurement. 
\end{postulate}

 This postulate can be viewed as a generalization of the no-cloning theorem of quantum theory \cite{wootters1982single}, which states that an unknown pure state cannot be cloned. No cloneable information generalizes this to sets of pure states, whereby it is impossible to copy information as to which subspace of the state space the system is in. This reduces to the no-cloning theorem when the subspace is question is the subspace generate by a single pure state. Furthermore it implies an analogous generalized no-broadcasting theorem \cite{barnum1996noncommuting}.  Similar postulates such as \emph{information gain implies disturbance} \cite{pfister2013information} are also sufficient to discount classical degrees of freedom.

\section{Results} 

If we want to determine the reversible dynamics of theories without entanglement, first we must be careful to define what we mean by an interaction. For example, consider the state spaces $\mathcal{A}$ and $\mathcal{B}:=\mathcal{A}\boxtimes\mathcal{A}$. Then we have the composite state space, $\mathcal{S}:=\mathcal{A}\boxtimes \mathcal{B}=\mathcal{A}\boxtimes(\mathcal{A}\boxtimes\mathcal{A})$. Then, $T:=\text{SWAP}_{A,A}\otimes \mathds{1}_A$ is a valid reversible transformation on $\mathcal{S}$ but does not factorize as $T= t_A\otimes t_B$ where $t_A:\mathcal{A}\to\mathcal{A}$ and $t_B:\mathcal{B}\to\mathcal{B}$ and so by definition \ref{def:RI} is a non-trivial reversible interaction. However, as this is really just swapping subsystems it does not generate any correlations between the systems and so we do not consider this to be a `genuine' interaction.

To eliminate such examples we limit ourselves to considering \emph{locally reversible interactions} as defined below.

\begin{definition}\label{def:LRI}A locally reversible interaction $T$ is one that satisfies the following:
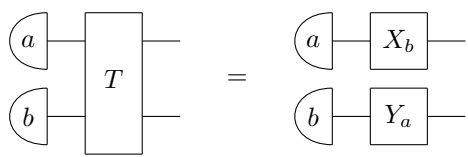
\[\input{LRI.tex}\]
where $T$, $X_b$ and $Y_a$ are reversible transformations.
\end{definition}
Such a transformation is a \emph{trivial} interaction if $\forall b\ \ X_b=X$ and $\forall a\ \ Y_a =Y$, for some reversible $X$ and $Y$. Trivial interactions generate no correlations. Examples of locally reversible interactions are the classical computational gates, or a quantum CNOT gate acting on the computational basis states. 

Note that as we are considering the minimal tensor product then this suffices to define the transformation $T$ as the theory must be tomographically local \cite{barnum2007generalized}. Moreover, if we relax the condition that $X_b$ and $Y_a$ are reversible then any bipartite transformation can be written in this way -- it is this reversibility constraint that rules out cases such as the above `partial swap' example.

We will now derive some consequences of the existence of these interactions, which will lead on to a proof of our main result for this section: the existence of a non-trivial locally reversible interaction implies classicality of the state space (in the sense of definition \ref{def:classicality}).

\begin{lemma} Local reversible interactions allow for the construction of `partial broadcasters' \cite{selby2016leaks}.
\end{lemma}
\proof We define a partial broadcaster, $B$, as a transformation satisfying the following equation:
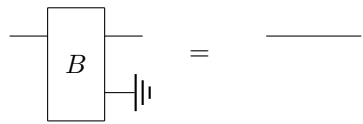
\[\input{leak.tex}\]
The standard broadcasting map additionally satisfies, an equivalent equation but with the trace on the top system.

We can define the following transformations using a local reversible interaction $T$:
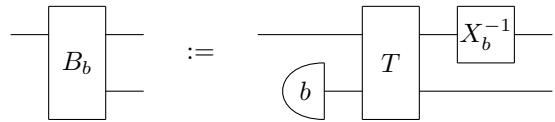
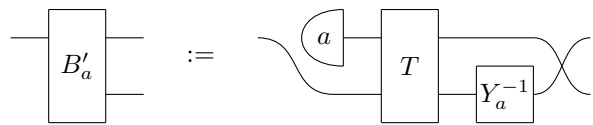
\[\input{Bb.tex}\]
and
\[\input{Ba.tex}\]
It is simple to check that these satisfy the defining equation for a partial broadcaster:
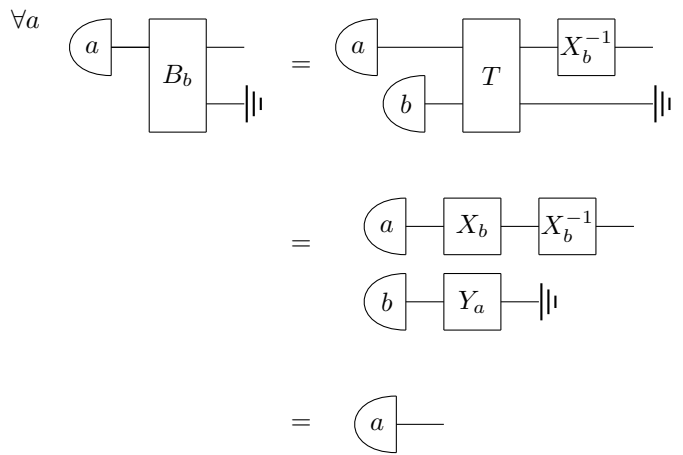
\[\input{broadcastingProof1.tex}\]
\[\input{broadcastingProof2.tex}\]
Similarly we can check that this equation is satisfied for $B'_a$.
\endproof

Note that we can have trivial partial broadcasters, these are of the form,
\[\input{TrivialBroadcaster.tex}\]
where $s$ is a normalised state. Note also that trivial locally reversible interactions give rise to only trivial partial broadcasters.

\begin{lemma}Partial broadcasters allow for non-disturbing measurements. 
\end{lemma}
\proof The defining equation for a partial broadcaster in theories where composition is given by the minimal tensor product implies the following:
\[\input{NDMProof1.tex}\]
where $f$ is some function from the set of pure states in $\mathcal{A}$ to the set of pure states in $\mathcal{B}$. Then we can construct a set of non-disturbing measurements $\{M_e\}$ as:
\[\input{NDMProof2.tex}\]
where these are non-disturbing as they satisfy,
\[\input{NDMProof3.tex}\]
\endproof

Note again that we can have the trivial case where the non-disturbing measurement has only a single outcome, i.e. such that $\forall s\ \ f(s)=c$, these are just a transformation proportional to the identity channel. This `decomposes' the state space in a trivial way, i.e. into a single component, and so does not lead to the state space having a classical degree of freedom. Any other non-disturbing measurement leads to a non-trivial decomposition.

We can then use the following result of Barnum  \emph{et al}:
\begin{lemma}\label{barnum lemma}
The existence of non-trivial non-disturbing measurements implies decomposability of the state space.
\end{lemma}
\proof Theorem 5 of \cite{barrett2007information}.
\endproof

Note that every theory has a trivial non-disturbing measurement, that is, anything proportional to the identity channel. This `decomposes' the state space in a trivial way, i.e. into a single component, and so does not lead to the state space having a classical degree of freedom. Any other non-disturbing measurement leads to a non-trivial decomposition. From Lemma \ref{barnum lemma} it is straightforward to prove the following:

\begin{theorem}For theories without entanglement (def. \ref{def:Entanglement}) with entirely non-classical state spaces (def. \ref{def:classicality}), all local reversible interactions are trivial (def. \ref{def:LRI}).
\end{theorem}
\proof
As we are assuming that the state spaces are entirely non-classical, any non-disturbing measurement constructed from $T$ must be proportional to the identity.

This means that the any `$f$'  obtained from a non-disturbing measurement must be a constant function, i.e. $\forall s\ f(s)=c$ for some fixed state $c$.

If we consider how $f$ is defined for a non-disturbing measurement constructed from $T$, it is defined by,
\[\input{LRIProof1.tex}\]
or,
\[\input{LRIProof2.tex}\]
where which $a$, $b$, $X\_$ or $Y\_$ is used depends on which partial broadcaster we construct.

However, non-classicality of the state space requires that $f(s)$ is constant for all possible broadcasters (else we get a non-trivial non-disturbing measurement and a non-trivial decomposition), and so, $\forall s\ X_s = X$ and $\forall s\ Y_s = Y$ where $X$ and $Y$ are some fixed reversible transformations.

Therefore,
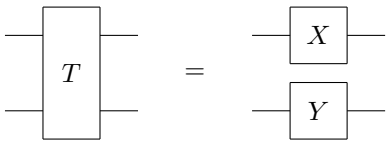
\[\input{LRIProof3.tex}\]
and so $T$ is a trivial interaction.
\endproof
If we relax the constraint of non-classicality, and instead allow for decomposable state spaces which encode classical degrees of freedom, we find that although reversible interactions exist they are conditional on these classical degrees of freedom only.\\

\begin{theorem}Interactions between systems without entanglement are conditional transformations on classical degrees of freedom
\end{theorem}

\begin{proof}
See Appendix A
\end{proof}

It is important to note that, although we make reference to non-disturbing measurements, this work does not require the commonly made assumption of the \emph{no-restriction hypothesis}, which asserts that any mathematically valid measurement can be physically realized. Moreover, we also are not using the frequently used assumption of \emph{local tomography}, which asserts that composite states can be completely characterized by their local measurement statistics. In summary, we find that if we assume the following postulates: 
\begin{itemize}
\item[1.] \emph{non-classicality}: the state space does not contain a classical degree of freedom;
\item[2.] \emph{no entanglement}: systems are composed under the minimal tensor product $\boxtimes$;
\item[3.] \emph{reversible interactions}: systems can become correlated through locally reversible interactions;
\end{itemize}
we reach a contradiction. By making use of a broad notion of classicality, we can arrive at two different sets of physical postulates that predict the existence of entangled states. Firstly, there is \emph{transitivity} (Postulate 1). Transitivity excludes the existence of classical degrees of freedom from non-classical systems, as for transitive state spaces these always manifest as classical ancillas. Therefore by Theorem 3 all reversible interactions between transitive systems are between a system and a classical system, and there are no reversibly interactions between non-classical systems. Secondly, \emph{no-cloning / broadcasting} (Postulate 2) immediately precludes the existence of classical degrees of freedom, and therefore any systems obeying a no-cloning / broadcasting postulate cannot interact reversibly with non-entangling transformations.

\section{Discussion}

In this article we have show that for systems to reversibly interact the interaction must generate entanglement or be an interaction of classical degrees of freedom. Furthermore, we have presented a set of physical postulates that are sufficient to imply the existence of entangled states by excluding classical degrees of freedom -- namely that the systems can contain no cloneable / broadcastable information or that the systems are transitive. To emphasis the power of this result, consider a physicist with no knowledge of quantum theory who observes systems with some non-classical properties such as exhibiting interference in a Young's double slit experiment. If the physicist observed that these systems can reversibly interact and their local state spaces can be reversibly explored, then he or she would be able to deduce the existence of entangled states. This implies that entanglement, far from being a manifestly quantum phenomena, is an inevitable feature of any reasonable non-classical theory of nature. 

The existence of reversible interactions in a key feature of both classical and quantum theory. For instance the second law of thermodynamics is contingent on the assumption of the reversibility of interactions \cite{bennett2003notes,alhambra2016second,richens2016quantum}. Previous work has shown reversibility to be an important property in determining the physical limitations of a theory. For example it is required to ensure the impossibility of deleting \cite{pati2000impossibility} or cloning an unknown quantum state \cite{wootters1982single} (and a weaker classical version \cite{daffertshofer2002classical}), which together establish information as something analogous to a conserved quantity.  

In classical and quantum theory not only do reversible interactions exist, but the theories are transitive, which means that you can reversibly explore all of the state space {(on both a local and global level)}. Indeed any theory whereby the local state space is defined by the action of some dynamical group (be it the permutation group in classical theory of the unitary group in quantum theory) obeys transitivity by construction. Therefore one could argue that transitivity should really be taken as part of the definition of a state space in theories where, fundamentally, all dynamics are reversible. For example if some region of the state space is not reachable under any of the local or multipartite dynamics then it should not be considered as part of the state space. If we use this definition of the state space, the only systems with classical degrees of freedom are classical simplices, and our main result becomes that non-classical systems cannot reversibly interact without entanglement.  

It would be of interest to determine if, perhaps by introducing additional physical postulates, it would be possible to derive the existence of non-local correlations that violate Bell inequalities. The existence of entangled states is in general a necessary but insufficient condition for observing violations of Bell inequalities. For example, entangled states are present in the local theory of Spekken's toy model \cite{spekkens2007evidence}. However, it has been shown that all entangled states in quantum theory display some hidden non-locality \cite{buscemi2012all,masanes2008all}. By determining the additional structure present in quantum theory that gives this correspondence between entanglement and non-locality, it could be possible to derive the violation of Bell inequalities from purely physical postulates. Given the surprising simplicity of the postulates presented in this paper that result in entanglement, it is plausible that the physical postulates that give rise to Bell non-locality are similarly mundane.

\noindent \textbf{Acknowledgments:} The authors would like to thank Lluis Masanes and Ciar\'{a}n Lee for useful discussion. JR and JHS are supported by EPSRC through the Controlled Quantum Dynamics Centre for Doctoral Training.
\bibliographystyle{apsrev4-1}
\bibliography{bibliography}

\section{Appendices}
\subsection{Mathematical background}

As mentioned above, we associate each system with a state space which is defined as 
a compact, convex subset of a real, finite-dimensional vector space. That is to say, a closed and bounded set of vectors in a real, finite-dimensional vector space such that if $v_1$ and $v_2$ are inside the set then for $p\in[0,1]$, $pv_1+(1-p)v_2$ is also in the set.

Given two such state spaces $\mathcal{A}$ and $\mathcal{B}$ there are two constructions of composite state spaces which are important for the derivation of our results. 

\begin{definition}Minimal tensor product, $\mathcal{A}\boxtimes \mathcal{B}$:
\[\mathcal{A}\boxtimes \mathcal{B}:=\text{Conv}\left(\left\{a\otimes b \middle| a\in \mathcal{A}, b\in\mathcal{B} \right\} \right)\]
\end{definition}
\begin{definition}Direct sum, $\mathcal{A}\oplus \mathcal{B}$:
\[\mathcal{A}\oplus \mathcal{B}:=\text{Conv}\left(\left\{a\oplus {\bf 0}, {\bf 0}\oplus b \middle| a\in \mathcal{A}, b\in\mathcal{B} \right\} \right)\]
\end{definition}

An important result is the following distributive law for the direct sum and minimal tensor product.

\begin{lemma}
The direct sum $\oplus$ distributes over the minimal tensor product $\boxtimes$
\end{lemma}
\begin{proof}
Consider three state spaces $\mathcal{A}$, $\mathcal B$ and $\mathcal C$. We want to show that,
\[\mathcal A\boxtimes (B\oplus C) = (A\boxtimes B)\oplus(A\boxtimes C)\]
This follows immediately from writing out both sides using the definitions of $\boxtimes$ and $\oplus$, using distributivity of the direct sum and tensor product on the individual vectors, and noting that ${\bf 0} \otimes s = {\bf 0}\otimes {\bf 0}$.
\end{proof}

Another important structure associated to the state space is the face lattice.
We say that a set $f$ is a \emph{face} of the state space $\mathcal{S}$ if $f$ is itself a convex set such that for any state $s\in f$, $v_i\in \mathcal{S}$ and $p\in[0,1]$ that,
 \[s=p v_1 + (1-p) v_2 \implies v_1,v_2 \in f\]
We denote by $\mathsf{Face}(\mathcal{S})$ the set of faces of $\mathcal{S}$. A vertex $v$ of $\mathcal{S}$ is a zero-dimensional face of $\mathcal{S}$, and we denote by $\mathsf{Vertex}(\mathcal{S})$ the set of vertices of $\mathcal{S}$. Faces of a convex set can be ordered by subset inclusion, which gives them the structure of a lattice. The join $f_1 \lor f_2$ of two faces in this lattice is particularly important for this work and can be defined as the minimal face that contains both $f_1$ and $f_2$.

It will be useful to consider how reversible transformations of the convex set induce a corresponding transformation of the face lattice. Specifically it induces a lattice automorphism, but the following three lemmas will be of particular use for us.

\begin{lemma}\label{lem:LatMorph2} $T(f)=f'\cong f$ i.e. faces are mapped to isomorphic faces. \end{lemma}
\proof Consider some $s'\in f'$ then $f'$ is a face if for any decomposition, $s'=p s'_1 +(1-p)s'_2$, $s'_i$ are also in the set $f'$. Note that $s'=T(s)$ for some $s\in f$ and as $T$ is reversible, this means that $s=T^{-1}(s')=T^{-1}(p s'_1 +(1-p)s'_2)$ which by linearity of $T^{-1}$ implies that $s=pT^{-1}(s'_1)+(1-p)T^{-1}(s'_2)$. This provides a decomposition of $s$ and as $f$ is a face this means that $T^{-1}(s'_i)\in f$ and so $s'_i$ are both in $f'$. Therefore $f'$ is a face. It is clearly isomorphic to $f$ as $T$ provides the isomorphism. \endproof

\begin{lemma}\label{lem:LatAutoMorph}A reversible transformation $T$ on $\mathcal{S}$ induces an automorphism of the face lattice $\mathsf{Face}(\mathcal{S})$.
\end{lemma}
\proof Lemma \ref{lem:LatMorph2} shows that faces are mapped to faces, then reversibility of $T$ implies that this mapping of faces must be 1 to 1 and hence induces a lattice automorphism.\endproof
An immediate corollary of this is that:
\begin{corollary}\label{lem:LatMorph1}
$T(s_1\lor s_2) =T(s_1)\lor T(s_2)$
\end{corollary}

\section{Appendix A}

\section{Transitive state spaces with classical degree of freedom}

\noindent \emph{Theorem 1:} Transitive state spaces with a classical degree of freedom have a classical subsystem, 
\[\mathcal{S}=\mathcal{A}\oplus \mathcal{B} \implies \exists N\geq 2\ \&\ \mathcal{C} \text{ such that } \mathcal{S}=\Delta_N\boxtimes \mathcal{C}\]

\proof
First note that one can always decompose a state space, $\mathcal{S}$ into irreducible components, $A_i$, as
\[\mathcal{S}=\bigoplus_i A_i\]
where irreducibility of $A_i$ implies that they cannot be further decomposed with respect to $\oplus$.

Next note that the faces of $\mathcal{S}$ are all of the form,
\[\bigoplus_i f_i\] where $f_i\in \mathsf{Face}(A_i)$ \cite{barker1978perfect}. Therefore the only irreducible faces are those of the form,
\[{\bf 0} \oplus\dots\oplus {\bf 0}\oplus f_i\oplus {\bf 0}\oplus \dots \oplus {\bf 0}\]
where $f_i$ is an irreducible face of $A_i$. 

As vertices are trivially irreducible, each vertex of $\mathcal{S}$ corresponds to one of the vertices of one of the $A_i$. Therefore, consider two vertices $v_i$, belonging to distinct $A_i$. As the $A_i$ are each irreducible (by assumption), then the two vertices of $\mathcal{S}$ corresponding to these $v_i$ belong to a maximal irreducible face isomorphic to their respective $A_i$.

Lemma \ref{lem:LatAutoMorph} implies that, reversible transformations preserve the set of faces that a vertex belongs to. Therefore the maximal irreducible face for each must be isomorphic, and hence, $A_i\cong A_j\ \ \forall i,j$.

 Therefore $\mathcal{S}=\bigoplus_{i=1}^n A$ for $A\cong A_j$. Now note that the simplex state space $\Delta_m=\bigoplus_{i=1}^{m+1}p$ where $p$ is a point state space. Additionally note that $B\boxtimes p\cong B$ for any state space $B$. Therefore, 
\[\mathcal{S}=\bigoplus_{i=1}^n A = \bigoplus_{i=1}^n A\boxtimes p = A\boxtimes \bigoplus_{i=1}^n p= A\boxtimes \Delta_{n-1}\]
where the third equality uses the distributivity of $\boxtimes$ over $\oplus$. 

\endproof
\begin{theorem}
Interactions between systems without entanglement are conditional transformations on those classical degrees of freedom

\end{theorem}

\begin{proof}

Consider the state space $\mathcal A \boxtimes \mathcal B$ where $\mathcal A = \bigoplus\limits_{i=1}^{n_A} \mathcal A_i$ where $\mathcal A_i$ are indecomposable and similarly $\mathcal B = \bigoplus\limits_{j=1}^{n_B} \mathcal B_j$ where all $\mathcal B_j$ are indecomposable (note any decomposable finite dimensional cone can be decomposed in the way \cite{barker1978perfect}). 
Consider a reversible transformation $T$ on this state space (as defined in Definition 4), 
\[T(\mathcal A\boxtimes \mathcal{B})=T(\bigoplus\limits_{i=1}^{n_A} \mathcal A_i\boxtimes \bigoplus\limits_{j=1}^{n_B} \mathcal B_j)=T(\bigoplus\limits_{i=1,j=1}^{n_A,n_B}\mathcal A_i \boxtimes \mathcal B_j)\]
the second equality uses distributivity of $\oplus$ over $\boxtimes$. As all $\mathcal A_i\boxtimes \mathcal B_j$ are indecomposable, they contain no classical degrees of freedom and therefore all locally reversible operations on them are trivial. Therefore Theorem 2 implies that 
\[T(\mathcal A_i\boxtimes \mathcal B_j)=X_i\otimes Y_j (\mathcal A_i \boxtimes \mathcal B_j) \]
and therefore 
\[ T(\mathcal A \boxtimes \mathcal B) = \bigoplus\limits_{i=1,j=1}^{n_A,n_B}T_i\otimes T_j(\mathcal A_i \boxtimes \mathcal B_j) \]
where $(i,j)$ index the classical degrees of freedom of $\mathcal A \boxtimes \mathcal B$. Therefore we see that, as $T$ factorizes over $\oplus$, the interaction is conditional on the classical degrees of freedom $(i,j)$ alone. 
\end{proof}

\end{document}

%% file: process.tex
\begin{tikzpicture}[scale=0.5]
	\begin{pgfonlayer}{nodelayer}
		\node [style=none] (0) at (1, 1.25) {};
		\node [style=none] (1) at (1, -1.25) {};
		\node [style=none] (2) at (2.5, -1.25) {};
		\node [style=none] (3) at (2.5, 1.25) {};
		\node [style=none] (4) at (1.750001, -0) {$f$};
		\node [style=none] (5) at (0, -0) {};
		\node [style=none] (6) at (1, -0) {};
		\node [style=none] (7) at (2.5, 0.5000001) {};
		\node [style=none] (8) at (3.5, 0.5000001) {};
		\node [style=none] (9) at (2.5, -0.5000001) {};
		\node [style=none] (10) at (3.5, -0.5000001) {};
		\node [style=none] (11) at (0.2499996, -0.5000001) {$A$};
		\node [style=none] (12) at (3.25, -1) {$B$};
		\node [style=none] (13) at (3.25, 1) {$A$};
		\node [style=none] (14) at (1.750001, 1.5) {};
	\end{pgfonlayer}
	\begin{pgfonlayer}{edgelayer}
		\draw (5.center) to (6.center);
		\draw (0.center) to (1.center);
		\draw (1.center) to (2.center);
		\draw (2.center) to (3.center);
		\draw (3.center) to (0.center);
		\draw (7.center) to (8.center);
		\draw (9.center) to (10.center);
	\end{pgfonlayer}
\end{tikzpicture}

%% file: state.tex
\begin{tikzpicture}[scale=0.5]
	\begin{pgfonlayer}{nodelayer}
		\node [style=none] (0) at (0, 1.25) {};
		\node [style=none] (1) at (0, -0.25) {};
		\node [style=none] (2) at (0, 0.5) {};
		\node [style=none] (3) at (1, 0.5) {};
		\node [style=none] (4) at (-0.5, 0.5) {$s$};
		\node [style=none] (5) at (0.75, 1) {$A$};
		\node [style=none] (6) at (0, 1.75) {};
		\node [style=none] (7) at (0, -0.7500002) {};
	\end{pgfonlayer}
	\begin{pgfonlayer}{edgelayer}
		\draw [bend right=90, looseness=2.25] (0.center) to (1.center);
		\draw (0.center) to (1.center);
		\draw (2.center) to (3.center);
	\end{pgfonlayer}
\end{tikzpicture}

%% file: effect.tex
\begin{tikzpicture}[scale=0.5]
		\node  (0) at (0, 0.5) {};
		\node  (1) at (1, 0.5) {};
		\node  (2) at (0, -0.5) {};
		\node  (3) at (1, -0.5) {};
		\node  (4) at (1, -1) {};
		\node  (5) at (1, 1) {};
		\node  (6) at (0.25, 1) {$A$};
		\node  (7) at (0.25, -1) {$B$};
		\node  (8) at (1.5, -0) {$e$};
		\draw [bend left=90, looseness=1.75] (5.center) to (4.center);
		\draw (5.center) to (4.center);
		\draw (2.center) to (3.center);
		\draw (0.center) to (1.center);
\end{tikzpicture}

%% file: experiment.tex
\begin{tikzpicture}[scale=0.5]
	\begin{pgfonlayer}{nodelayer}
		\node [style=none] (0) at (-0.2499996, 2.25) {};
		\node [style=none] (1) at (-0.2499996, -1.25) {};
		\node [style=none] (2) at (1.25, -1.25) {};
		\node [style=none] (3) at (1.25, 2.25) {};
		\node [style=none] (4) at (0.5, 0.5) {$f$};
		\node [style=none] (5) at (-1.25, 0.5000001) {};
		\node [style=none] (6) at (-0.2499996, 0.5000001) {};
		\node [style=none] (7) at (1.25, 1.5) {};
		\node [style=none] (8) at (2.25, 1.5) {};
		\node [style=none] (9) at (1.25, -0.5000001) {};
		\node [style=none] (10) at (2.25, -0.5000001) {};
		\node [style=none] (11) at (-0.7499997, -0) {$A$};
		\node [style=none] (12) at (1.750001, -1) {$B$};
		\node [style=none] (13) at (1.750001, 2) {$C$};
		\node [style=none] (14) at (-1.25, 3) {};
		\node [style=none] (15) at (2.25, 3) {};
		\node [style=none] (16) at (2.25, 3.75) {};
		\node [style=none] (17) at (2.25, 0.7500002) {};
		\node [style=none] (18) at (-1.25, 3.75) {};
		\node [style=none] (19) at (-1.25, -0.2500001) {};
		\node [style=none] (20) at (2.25, 0.2500001) {};
		\node [style=none] (21) at (2.25, -1.25) {};
		\node [style=none] (22) at (0.5, 3.5) {$A$};
		\node [style=none] (23) at (-1.75, 1.75) {$s$};
		\node [style=none] (24) at (2.749999, 2.25) {$e_1$};
		\node [style=none] (25) at (2.749999, -0.5000001) {$e_2$};
		\node [style=none] (26) at (6, 1.5) {$\sim\ p(s,f,e_1,e_2)$};
	\end{pgfonlayer}
	\begin{pgfonlayer}{edgelayer}
		\draw (5.center) to (6.center);
		\draw (0.center) to (1.center);
		\draw (1.center) to (2.center);
		\draw (2.center) to (3.center);
		\draw (3.center) to (0.center);
		\draw (7.center) to (8.center);
		\draw (9.center) to (10.center);
		\draw [bend right=90, looseness=1.00] (18.center) to (19.center);
		\draw (18.center) to (19.center);
		\draw (14.center) to (15.center);
		\draw [bend left=90, looseness=1.25] (16.center) to (17.center);
		\draw (16.center) to (17.center);
		\draw [bend left=90, looseness=2.00] (20.center) to (21.center);
		\draw (20.center) to (21.center);
	\end{pgfonlayer}
\end{tikzpicture}

%% file: trace.tex
\begin{tikzpicture}[scale=0.5]
	\begin{pgfonlayer}{nodelayer}
		\node [style=none] (0) at (0, -0) {};
		\node [style=trace] (1) at (0.7499999, -0) {};
	\end{pgfonlayer}
	\begin{pgfonlayer}{edgelayer}
		\draw (0.center) to (1);
	\end{pgfonlayer}
\end{tikzpicture}

%% file: entanglement.tex
\begin{tikzpicture}[scale=0.5]
		\node  (0) at (-0.75, -1) {};
		\node  (1) at (0, -1) {};
		\node  (2) at (-0.75, 1) {};
		\node  (3) at (0, 1) {};
		\node  (4) at (-0.75, 1.5) {};
		\node  (5) at (-0.75, -1.5) {};
		\node  (6) at (-1.25, -0) {$s$};
		\node  (7) at (1, -0) {$=$};
		\node  (8) at (3, 1.75) {};
		\node  (9) at (3, 1) {};
		\node  (10) at (3, 0.25) {};
		\node  (11) at (3, -0.25) {};
		\node  (12) at (3, -1) {};
		\node  (13) at (3, -1.75) {};
		\node  (14) at (3.75, 1) {};
		\node  (15) at (3.75, -1) {};
		\node  (16) at (2.5, 1) {$s_1$};
		\node  (17) at (2.5, -1) {$s_2$};
        \node (18) at (-5,0){$\nexists s_1,s_2 \text{ such that }$};
		\draw (0.center) to (1.center);
		\draw [bend right=90, looseness=1.25] (4.center) to (5.center);
		\draw (4.center) to (5.center);
		\draw (2.center) to (3.center);
		\draw [bend right=90, looseness=2.25] (8.center) to (10.center);
		\draw [bend right=90, looseness=2.25] (11.center) to (13.center);
		\draw (12.center) to (15.center);
		\draw (9.center) to (14.center);
		\draw (8.center) to (10.center);
		\draw (11.center) to (13.center);
\end{tikzpicture}

%% file: interaction1.tex
\begin{tikzpicture}[scale=0.5]
	\begin{pgfonlayer}{nodelayer}
		\node [style=none] (0) at (-0.75, -0) {};
		\node [style=none] (1) at (0, -0) {};
		\node [style=none] (2) at (-0.75, 2) {};
		\node [style=none] (3) at (0, 2) {};
		\node [style=none] (4) at (-0.75, -0.5) {};
		\node [style=none] (5) at (-1.5, 1) {$T$};
		\node [style=none] (6) at (1, 1) {$=$};
		\node [style=none] (7) at (4.25, 2) {};
		\node [style=none] (8) at (4.25, 1.25) {};
		\node [style=none] (9) at (4.25, -0) {};
		\node [style=none] (10) at (4.25, -0.75) {};
		\node [style=none] (11) at (5, 2) {};
		\node [style=none] (12) at (5, -0) {};
		\node [style=none] (13) at (3.5, 2) {$T_1$};
		\node [style=none] (14) at (3.5, -0) {$T_2$};
		\node [style=none] (15) at (-0.75, 2.5) {};
		\node [style=none] (16) at (-2.25, 2.5) {};
		\node [style=none] (17) at (-2.25, 2) {};
		\node [style=none] (18) at (-3, 2) {};
		\node [style=none] (19) at (-3, -0) {};
		\node [style=none] (20) at (-2.25, -0) {};
		\node [style=none] (21) at (-2.25, -0.5) {};
		\node [style=none] (22) at (4.25, 0.75) {};
		\node [style=none] (23) at (2.75, 0.75) {};
		\node [style=none] (24) at (2.75, -0.75) {};
		\node [style=none] (25) at (2.75, 1.25) {};
		\node [style=none] (26) at (2.75, 2.75) {};
		\node [style=none] (27) at (4.25, 2.75) {};
		\node [style=none] (28) at (2.75, 2) {};
		\node [style=none] (29) at (2, 2) {};
		\node [style=none] (30) at (2, -0) {};
		\node [style=none] (31) at (2.75, -0) {};
		\node [style=none] (32) at (-7, 1) {$\nexists T_1\ \&\ T_2 \text{ such that }$};
	\end{pgfonlayer}
	\begin{pgfonlayer}{edgelayer}
		\draw (0.center) to (1.center);
		\draw (2.center) to (3.center);
		\draw (9.center) to (12.center);
		\draw (7.center) to (11.center);
		\draw (18.center) to (17.center);
		\draw (16.center) to (21.center);
		\draw (19.center) to (20.center);
		\draw (21.center) to (4.center);
		\draw (4.center) to (15.center);
		\draw (15.center) to (16.center);
		\draw (26.center) to (25.center);
		\draw (25.center) to (8.center);
		\draw (8.center) to (27.center);
		\draw (27.center) to (26.center);
		\draw (29.center) to (28.center);
		\draw (30.center) to (31.center);
		\draw (23.center) to (24.center);
		\draw (24.center) to (10.center);
		\draw (10.center) to (22.center);
		\draw (22.center) to (23.center);
	\end{pgfonlayer}
\end{tikzpicture}

%% file: LRI.tex
\begin{tikzpicture}[scale=0.5]
	\begin{pgfonlayer}{nodelayer}
		\node [style=none] (0) at (0, 2) {};
		\node [style=none] (1) at (0, 2.75) {};
		\node [style=none] (2) at (0, 1.25) {};
		\node [style=none] (3) at (0, 0.75) {};
		\node [style=none] (4) at (0, -0) {};
		\node [style=none] (5) at (0, -0.75) {};
		\node [style=none] (6) at (1, -0) {};
		\node [style=none] (7) at (1, 2) {};
		\node [style=none] (8) at (1, 2.75) {};
		\node [style=none] (9) at (1, -1) {};
		\node [style=none] (10) at (2.5, -1) {};
		\node [style=none] (11) at (2.5, 2.75) {};
		\node [style=none] (12) at (2.5, 2) {};
		\node [style=none] (13) at (2.5, -0) {};
		\node [style=none] (14) at (3.5, 2) {};
		\node [style=none] (15) at (3.5, -0) {};
		\node [style=none] (16) at (5, 1) {$=$};
		\node [style=none] (17) at (-0.5, 2) {$a$};
		\node [style=none] (18) at (-0.5, -0) {$b$};
		\node [style=none] (19) at (1.75, 1) {$T$};
		\node [style=none] (20) at (11, 2) {};
		\node [style=none] (21) at (8.5, -0) {};
		\node [style=none] (22) at (7.5, 1.25) {};
		\node [style=none] (23) at (7.5, -0.75) {};
		\node [style=none] (24) at (7.5, 2) {};
		\node [style=none] (25) at (8.5, 2.75) {};
		\node [style=none] (26) at (7, 2) {$a$};
		\node [style=none] (27) at (10, -0) {};
		\node [style=none] (28) at (8.5, 2) {};
		\node [style=none] (29) at (10, 2) {};
		\node [style=none] (30) at (10, 2.75) {};
		\node [style=none] (31) at (7.5, -0) {};
		\node [style=none] (32) at (11, -0) {};
		\node [style=none] (33) at (7.5, 2.75) {};
		\node [style=none] (34) at (7, -0) {$b$};
		\node [style=none] (35) at (7.5, 0.75) {};
		\node [style=none] (36) at (0, 4) {$\forall a,\ b\quad \exists X_b,\ Y_a \text{ such that }$};
		\node [style=none] (37) at (8.5, 1.25) {};
		\node [style=none] (38) at (10, 1.25) {};
		\node [style=none] (39) at (10, 0.75) {};
		\node [style=none] (40) at (8.5, 0.75) {};
		\node [style=none] (41) at (8.5, -0.75) {};
		\node [style=none] (42) at (10, -0.75) {};
		\node [style=none] (43) at (9.25, 2) {$X_b$};
		\node [style=none] (44) at (9.25, -0) {$Y_a$};
	\end{pgfonlayer}
	\begin{pgfonlayer}{edgelayer}
		\draw (8.center) to (9.center);
		\draw (9.center) to (10.center);
		\draw (10.center) to (11.center);
		\draw (11.center) to (8.center);
		\draw (12.center) to (14.center);
		\draw (13.center) to (15.center);
		\draw (4.center) to (6.center);
		\draw [bend right=90, looseness=2.25] (3.center) to (5.center);
		\draw (3.center) to (5.center);
		\draw [bend right=90, looseness=2.25] (1.center) to (2.center);
		\draw (1.center) to (2.center);
		\draw (0.center) to (7.center);
		\draw (30.center) to (25.center);
		\draw (29.center) to (20.center);
		\draw (27.center) to (32.center);
		\draw (31.center) to (21.center);
		\draw [bend right=90, looseness=2.25] (35.center) to (23.center);
		\draw (35.center) to (23.center);
		\draw [bend right=90, looseness=2.25] (33.center) to (22.center);
		\draw (33.center) to (22.center);
		\draw (24.center) to (28.center);
		\draw (25.center) to (37.center);
		\draw (37.center) to (38.center);
		\draw (38.center) to (30.center);
		\draw (40.center) to (39.center);
		\draw (39.center) to (42.center);
		\draw (42.center) to (41.center);
		\draw (41.center) to (40.center);
	\end{pgfonlayer}
\end{tikzpicture}

%% file: leak.tex
\begin{tikzpicture}[scale=0.5]
	\begin{pgfonlayer}{nodelayer}
		\node [style=none] (0) at (0, 1) {};
		\node [style=none] (1) at (1, -0.5000001) {};
		\node [style=none] (2) at (1, 1) {};
		\node [style=none] (3) at (1, 1.75) {};
		\node [style=none] (4) at (1, -1.25) {};
		\node [style=none] (5) at (2.5, -1.25) {};
		\node [style=none] (6) at (2.5, 1.75) {};
		\node [style=none] (7) at (2.5, 1) {};
		\node [style=none] (8) at (2.5, -0.5000001) {};
		\node [style=none] (9) at (3.5, 1) {};
		\node [style=trace] (10) at (3.5, -0.5000001) {};
		\node [style=none] (11) at (5, 0.5000001) {$=$};
		\node [style=none] (12) at (1.75, 0.2500001) {$B$};
		\node [style=none] (13) at (6.75, 1) {};
		\node [style=none] (14) at (9.25, 1) {};
	\end{pgfonlayer}
	\begin{pgfonlayer}{edgelayer}
		\draw (3.center) to (4.center);
		\draw (4.center) to (5.center);
		\draw (5.center) to (6.center);
		\draw (6.center) to (3.center);
		\draw (7.center) to (9.center);
		\draw (8.center) to (10);
		\draw (0.center) to (2.center);
		\draw (13.center) to (14.center);
	\end{pgfonlayer}
\end{tikzpicture}

%% file: Bb.tex
\begin{tikzpicture}[scale=0.5]
	\begin{pgfonlayer}{nodelayer}
		\node [style=none] (0) at (0, 1) {};
		\node [style=none] (1) at (1, -0.5000001) {};
		\node [style=none] (2) at (1, 1) {};
		\node [style=none] (3) at (1, 1.75) {};
		\node [style=none] (4) at (1, -1.25) {};
		\node [style=none] (5) at (2.5, -1.25) {};
		\node [style=none] (6) at (2.5, 1.75) {};
		\node [style=none] (7) at (2.5, 1) {};
		\node [style=none] (8) at (2.5, -0.5000001) {};
		\node [style=none] (9) at (3.5, 1) {};
		\node [style=none] (10) at (5, 0.5000001) {$:=$};
		\node [style=none] (11) at (1.75, 0.2500001) {$B_b$};
		\node [style=none] (12) at (3.5, -0.5000001) {};
		\node [style=none] (13) at (10.75, 1) {};
		\node [style=none] (14) at (10.75, 1.75) {};
		\node [style=none] (15) at (6.5, 1) {};
		\node [style=none] (16) at (14.25, -0.5000001) {};
		\node [style=none] (17) at (10, 0.2500001) {$T$};
		\node [style=none] (18) at (9.25, -1.25) {};
		\node [style=none] (19) at (11.75, 1) {};
		\node [style=none] (20) at (10.75, -1.25) {};
		\node [style=none] (21) at (9.25, 1.75) {};
		\node [style=none] (22) at (10.75, -0.5000001) {};
		\node [style=none] (23) at (9.25, 1) {};
		\node [style=none] (24) at (9.25, -0.5000001) {};
		\node [style=none] (25) at (8.25, -0.5000001) {};
		\node [style=none] (26) at (8.25, 0.2500001) {};
		\node [style=none] (27) at (8.25, -1.25) {};
		\node [style=none] (28) at (7.75, -0.5000001) {$b$};
		\node [style=none] (29) at (11.75, 1.75) {};
		\node [style=none] (30) at (11.75, 0.2500001) {};
		\node [style=none] (31) at (13.25, 0.2500001) {};
		\node [style=none] (32) at (13.25, 1.75) {};
		\node [style=none] (33) at (12.5, 1) {$X_b^{-1}$};
		\node [style=none] (34) at (13.25, 1) {};
		\node [style=none] (35) at (14.25, 1) {};
	\end{pgfonlayer}
	\begin{pgfonlayer}{edgelayer}
		\draw (3.center) to (4.center);
		\draw (4.center) to (5.center);
		\draw (5.center) to (6.center);
		\draw (6.center) to (3.center);
		\draw (7.center) to (9.center);
		\draw (0.center) to (2.center);
		\draw (8.center) to (12.center);
		\draw (21.center) to (18.center);
		\draw (18.center) to (20.center);
		\draw (20.center) to (14.center);
		\draw (14.center) to (21.center);
		\draw (13.center) to (19.center);
		\draw (15.center) to (23.center);
		\draw (22.center) to (16.center);
		\draw [bend right=90, looseness=2.50] (26.center) to (27.center);
		\draw (25.center) to (24.center);
		\draw (26.center) to (27.center);
		\draw (29.center) to (30.center);
		\draw (30.center) to (31.center);
		\draw (31.center) to (32.center);
		\draw (32.center) to (29.center);
		\draw (34.center) to (35.center);
	\end{pgfonlayer}
\end{tikzpicture}

%% file: Ba.tex
\begin{tikzpicture}[scale=0.5]
	\begin{pgfonlayer}{nodelayer}
		\node [style=none] (0) at (0, 1) {};
		\node [style=none] (1) at (1, -0.5000001) {};
		\node [style=none] (2) at (1, 1) {};
		\node [style=none] (3) at (1, 1.75) {};
		\node [style=none] (4) at (1, -1.25) {};
		\node [style=none] (5) at (2.5, -1.25) {};
		\node [style=none] (6) at (2.5, 1.75) {};
		\node [style=none] (7) at (2.5, 1) {};
		\node [style=none] (8) at (2.5, -0.5000001) {};
		\node [style=none] (9) at (3.5, 1) {};
		\node [style=none] (10) at (5, 0.5000001) {$:=$};
		\node [style=none] (11) at (1.75, 0.2500001) {$B'_a$};
		\node [style=none] (12) at (3.5, -0.5000001) {};
		\node [style=none] (13) at (11.25, -0.5000001) {};
		\node [style=none] (14) at (11.25, 1.75) {};
		\node [style=none] (15) at (6.5, 1) {};
		\node [style=none] (16) at (13.75, 1) {};
		\node [style=none] (17) at (10.5, 0.2500001) {$T$};
		\node [style=none] (18) at (9.750001, -1.25) {};
		\node [style=none] (19) at (12.25, -0.5000001) {};
		\node [style=none] (20) at (11.25, -1.25) {};
		\node [style=none] (21) at (9.750001, 1.75) {};
		\node [style=none] (22) at (11.25, 1) {};
		\node [style=none] (23) at (8.499999, -0.5000001) {};
		\node [style=none] (24) at (9.750001, 1) {};
		\node [style=none] (25) at (8.750001, 1) {};
		\node [style=none] (26) at (8.750001, 1.75) {};
		\node [style=none] (27) at (8.750001, 0.2500001) {};
		\node [style=none] (28) at (8.25, 1) {$a$};
		\node [style=none] (29) at (12.25, 0.2500001) {};
		\node [style=none] (30) at (12.25, -1.25) {};
		\node [style=none] (31) at (13.75, -1.25) {};
		\node [style=none] (32) at (13.75, 0.2500001) {};
		\node [style=none] (33) at (13, -0.5000001) {$Y_a^{-1}$};
		\node [style=none] (34) at (13.75, -0.5000001) {};
		\node [style=none] (35) at (15.25, 1) {};
		\node [style=none] (36) at (15.25, -0.5000001) {};
		\node [style=none] (37) at (9.750001, -0.5000001) {};
	\end{pgfonlayer}
	\begin{pgfonlayer}{edgelayer}
		\draw (3.center) to (4.center);
		\draw (4.center) to (5.center);
		\draw (5.center) to (6.center);
		\draw (6.center) to (3.center);
		\draw (7.center) to (9.center);
		\draw (0.center) to (2.center);
		\draw (8.center) to (12.center);
		\draw (21.center) to (18.center);
		\draw (18.center) to (20.center);
		\draw (20.center) to (14.center);
		\draw (14.center) to (21.center);
		\draw (13.center) to (19.center);
		\draw [in=180, out=0, looseness=1.25] (15.center) to (23.center);
		\draw (22.center) to (16.center);
		\draw [bend right=90, looseness=2.50] (26.center) to (27.center);
		\draw (25.center) to (24.center);
		\draw (26.center) to (27.center);
		\draw (29.center) to (30.center);
		\draw (30.center) to (31.center);
		\draw (31.center) to (32.center);
		\draw (32.center) to (29.center);
		\draw [in=180, out=0, looseness=1.00] (34.center) to (35.center);
		\draw [in=180, out=0, looseness=1.00] (16.center) to (36.center);
		\draw (23.center) to (37.center);
	\end{pgfonlayer}
\end{tikzpicture}

%% file: broadcastingProof1.tex
\begin{tikzpicture}[scale=0.5]
	\begin{pgfonlayer}{nodelayer}
		\node [style=none] (0) at (0, 1) {};
		\node [style=none] (1) at (1, -0.5000001) {};
		\node [style=none] (2) at (1, 1) {};
		\node [style=none] (3) at (1, 1.75) {};
		\node [style=none] (4) at (1, -1.25) {};
		\node [style=none] (5) at (2.5, -1.25) {};
		\node [style=none] (6) at (2.5, 1.75) {};
		\node [style=none] (7) at (2.5, 1) {};
		\node [style=none] (8) at (2.5, -0.5000001) {};
		\node [style=none] (9) at (3.5, 1) {};
		\node [style=none] (10) at (5, 0.5000001) {$=$};
		\node [style=none] (11) at (1.75, 0.2500001) {$B_b$};
		\node [style=none] (12) at (3.5, -0.5000001) {};
		\node [style=none] (13) at (10.75, 1) {};
		\node [style=none] (14) at (10.75, 1.75) {};
		\node [style=none] (15) at (7, 1) {};
		\node [style=none] (16) at (14.25, -0.5000001) {};
		\node [style=none] (17) at (10, 0.2500001) {$T$};
		\node [style=none] (18) at (9.25, -1.25) {};
		\node [style=none] (19) at (11.75, 1) {};
		\node [style=none] (20) at (10.75, -1.25) {};
		\node [style=none] (21) at (9.25, 1.75) {};
		\node [style=none] (22) at (10.75, -0.5000001) {};
		\node [style=none] (23) at (9.25, 1) {};
		\node [style=none] (24) at (9.25, -0.5000001) {};
		\node [style=none] (25) at (8.25, -0.5000001) {};
		\node [style=none] (26) at (8.25, 0.2500001) {};
		\node [style=none] (27) at (8.25, -1.25) {};
		\node [style=none] (28) at (7.75, -0.5000001) {$b$};
		\node [style=none] (29) at (11.75, 1.75) {};
		\node [style=none] (30) at (11.75, 0.2500001) {};
		\node [style=none] (31) at (13.25, 0.2500001) {};
		\node [style=none] (32) at (13.25, 1.75) {};
		\node [style=none] (33) at (12.5, 1) {$X_b^{-1}$};
		\node [style=none] (34) at (13.25, 1) {};
		\node [style=none] (35) at (14.25, 1) {};
		\node [style=none] (36) at (0, 1.75) {};
		\node [style=none] (37) at (0, 0.2500001) {};
		\node [style=none] (38) at (0, 1) {};
		\node [style=none] (39) at (1, 1) {};
		\node [style=none] (40) at (-0.5000001, 1) {$a$};
		\node [style=none] (41) at (-2.25, 1.75) {$\forall a$};
		\node [style=none] (42) at (7, 1.75) {};
		\node [style=none] (43) at (7, 0.2500001) {};
		\node [style=none] (44) at (7, 1) {};
		\node [style=none] (45) at (7, 1) {};
		\node [style=none] (46) at (6.5, 1) {$a$};
		\node [style=none] (47) at (7.75, -3.75) {};
		\node [style=none] (48) at (7.75, -5) {};
		\node [style=none] (49) at (12.75, -3) {};
		\node [style=none] (50) at (7.25, -3.75) {$a$};
		\node [style=none] (51) at (5, -4.25) {$=$};
		\node [style=none] (52) at (7.75, -3) {};
		\node [style=none] (53) at (8.750001, -3.75) {};
		\node [style=none] (54) at (11.25, -3) {};
		\node [style=none] (55) at (12.75, -3.75) {};
		\node [style=none] (56) at (7.75, -3.75) {};
		\node [style=none] (57) at (7.75, -6.5) {};
		\node [style=none] (58) at (12, -3.75) {$X_b^{-1}$};
		\node [style=none] (59) at (10.25, -3.75) {};
		\node [style=none] (60) at (11.25, -3.75) {};
		\node [style=none] (61) at (7.75, -5.75) {};
		\node [style=none] (62) at (7.75, -3.75) {};
		\node [style=none] (63) at (12.75, -4.5) {};
		\node [style=none] (64) at (8.750001, -5.75) {};
		\node [style=none] (65) at (11.25, -4.5) {};
		\node [style=none] (66) at (11.25, -5.75) {};
		\node [style=none] (67) at (10.25, -5.75) {};
		\node [style=none] (68) at (13.75, -3.75) {};
		\node [style=none] (69) at (7.75, -4.5) {};
		\node [style=none] (70) at (7.25, -5.75) {$b$};
		\node [style=none] (71) at (10.25, -4.5) {};
		\node [style=none] (72) at (9.499999, -3.75) {$X_b$};
		\node [style=none] (73) at (8.750001, -3) {};
		\node [style=none] (74) at (10.25, -3) {};
		\node [style=none] (75) at (10.25, -3.75) {};
		\node [style=none] (76) at (8.750001, -4.5) {};
		\node [style=none] (77) at (8.750001, -3.75) {};
		\node [style=none] (78) at (10.25, -6.5) {};
		\node [style=none] (79) at (9.499999, -5.75) {$Y_a$};
		\node [style=none] (80) at (8.750001, -5) {};
		\node [style=none] (81) at (10.25, -5) {};
		\node [style=none] (82) at (10.25, -5.75) {};
		\node [style=none] (83) at (8.750001, -6.5) {};
		\node [style=none] (84) at (8.750001, -5.75) {};
		\node [style=trace] (85) at (3.75, -0.5000001) {};
		\node [style=trace] (86) at (14.5, -0.5000001) {};
		\node [style=trace] (87) at (11.5, -5.75) {};
		\node [style=none] (88) at (7.5, -8.25) {};
		\node [style=none] (89) at (7.5, -9.75) {};
		\node [style=none] (90) at (7, -9) {$a$};
		\node [style=none] (91) at (7.5, -9) {};
		\node [style=none] (92) at (7.5, -9) {};
		\node [style=none] (93) at (5, -9) {$=$};
		\node [style=none] (94) at (7.5, -9) {};
		\node [style=none] (95) at (8.750001, -9) {};
	\end{pgfonlayer}
	\begin{pgfonlayer}{edgelayer}
		\draw (3.center) to (4.center);
		\draw (4.center) to (5.center);
		\draw (5.center) to (6.center);
		\draw (6.center) to (3.center);
		\draw (7.center) to (9.center);
		\draw (0.center) to (2.center);
		\draw (8.center) to (12.center);
		\draw (21.center) to (18.center);
		\draw (18.center) to (20.center);
		\draw (20.center) to (14.center);
		\draw (14.center) to (21.center);
		\draw (13.center) to (19.center);
		\draw (15.center) to (23.center);
		\draw (22.center) to (16.center);
		\draw [bend right=90, looseness=2.50] (26.center) to (27.center);
		\draw (25.center) to (24.center);
		\draw (26.center) to (27.center);
		\draw (29.center) to (30.center);
		\draw (30.center) to (31.center);
		\draw (31.center) to (32.center);
		\draw (32.center) to (29.center);
		\draw (34.center) to (35.center);
		\draw [bend right=90, looseness=2.50] (36.center) to (37.center);
		\draw (38.center) to (39.center);
		\draw (36.center) to (37.center);
		\draw [bend right=90, looseness=2.50] (42.center) to (43.center);
		\draw (42.center) to (43.center);
		\draw (59.center) to (60.center);
		\draw (56.center) to (53.center);
		\draw (67.center) to (66.center);
		\draw [bend right=90, looseness=2.50] (48.center) to (57.center);
		\draw (61.center) to (64.center);
		\draw (48.center) to (57.center);
		\draw (54.center) to (65.center);
		\draw (65.center) to (63.center);
		\draw (63.center) to (49.center);
		\draw (49.center) to (54.center);
		\draw (55.center) to (68.center);
		\draw [bend right=90, looseness=2.50] (52.center) to (69.center);
		\draw (52.center) to (69.center);
		\draw (73.center) to (76.center);
		\draw (76.center) to (71.center);
		\draw (71.center) to (74.center);
		\draw (74.center) to (73.center);
		\draw (80.center) to (83.center);
		\draw (83.center) to (78.center);
		\draw (78.center) to (81.center);
		\draw (81.center) to (80.center);
		\draw [bend right=90, looseness=2.50] (88.center) to (89.center);
		\draw (88.center) to (89.center);
		\draw (91.center) to (95.center);
	\end{pgfonlayer}
\end{tikzpicture}

%% file: broadcastingProof2.tex
\begin{tikzpicture}[scale=0.5]
	\begin{pgfonlayer}{nodelayer}
		\node [style=none] (0) at (0, 1) {};
		\node [style=none] (1) at (1, -0.5000001) {};
		\node [style=none] (2) at (1, 1) {};
		\node [style=none] (3) at (1, 1.75) {};
		\node [style=none] (4) at (1, -1.25) {};
		\node [style=none] (5) at (2.5, -1.25) {};
		\node [style=none] (6) at (2.5, 1.75) {};
		\node [style=none] (7) at (2.5, 1) {};
		\node [style=none] (8) at (2.5, -0.5000001) {};
		\node [style=none] (9) at (3.5, 1) {};
		\node [style=trace] (10) at (3.5, -0.5000001) {};
		\node [style=none] (11) at (5, 0.5000001) {$=$};
		\node [style=none] (12) at (1.75, 0.2500001) {$B_b$};
		\node [style=none] (13) at (6.75, 1) {};
		\node [style=none] (14) at (9.25, 1) {};
		\node [style=none] (15) at (-1.5, 0.5000001) {$\implies$};
	\end{pgfonlayer}
	\begin{pgfonlayer}{edgelayer}
		\draw (3.center) to (4.center);
		\draw (4.center) to (5.center);
		\draw (5.center) to (6.center);
		\draw (6.center) to (3.center);
		\draw (7.center) to (9.center);
		\draw (8.center) to (10);
		\draw (0.center) to (2.center);
		\draw (13.center) to (14.center);
	\end{pgfonlayer}
\end{tikzpicture}

%% file: TrivialBroadcaster.tex
\begin{tikzpicture}[scale=0.5]
	\begin{pgfonlayer}{nodelayer}
		\node [style=none] (0) at (0, 1) {};
		\node [style=none] (1) at (1, -0.5000001) {};
		\node [style=none] (2) at (1, 1) {};
		\node [style=none] (3) at (1, 1.75) {};
		\node [style=none] (4) at (1, -1.25) {};
		\node [style=none] (5) at (2.5, -1.25) {};
		\node [style=none] (6) at (2.5, 1.75) {};
		\node [style=none] (7) at (2.5, 1) {};
		\node [style=none] (8) at (2.5, -0.5000001) {};
		\node [style=none] (9) at (3.5, 1) {};
		\node [style=none] (10) at (5.25, -0) {$=$};
		\node [style=none] (11) at (1.75, 0.2500001) {$B$};
		\node [style=none] (12) at (6.75, 1) {};
		\node [style=none] (13) at (9.25, 1) {};
		\node [style=none] (14) at (3.5, -0.5000001) {};
		\node [style=none] (15) at (9.25, -0.5000001) {};
		\node [style=none] (16) at (8.25, -0.5000001) {};
		\node [style=none] (17) at (8.25, 0.2500001) {};
		\node [style=none] (18) at (8.25, -1.25) {};
		\node [style=none] (19) at (7.75, -0.5000001) {$s$};
	\end{pgfonlayer}
	\begin{pgfonlayer}{edgelayer}
		\draw (3.center) to (4.center);
		\draw (4.center) to (5.center);
		\draw (5.center) to (6.center);
		\draw (6.center) to (3.center);
		\draw (7.center) to (9.center);
		\draw (0.center) to (2.center);
		\draw (12.center) to (13.center);
		\draw (8.center) to (14.center);
		\draw [bend right=90, looseness=2.25] (17.center) to (18.center);
		\draw (16.center) to (15.center);
		\draw (17.center) to (18.center);
	\end{pgfonlayer}
\end{tikzpicture}

%% file: NDMProof1.tex
\begin{tikzpicture}[scale=0.5]
	\begin{pgfonlayer}{nodelayer}
		\node [style=none] (0) at (0, 1) {};
		\node [style=none] (1) at (1, -0.5000001) {};
		\node [style=none] (2) at (1, 1) {};
		\node [style=none] (3) at (1, 1.75) {};
		\node [style=none] (4) at (1, -1.25) {};
		\node [style=none] (5) at (2.5, -1.25) {};
		\node [style=none] (6) at (2.5, 1.75) {};
		\node [style=none] (7) at (2.5, 1) {};
		\node [style=none] (8) at (2.5, -0.5000001) {};
		\node [style=none] (9) at (3.5, 1) {};
		\node [style=none] (10) at (5, 0.4999999) {$=$};
		\node [style=none] (11) at (1.75, 0.2500001) {$B$};
		\node [style=none] (12) at (3.5, -0.4999999) {};
		\node [style=none] (13) at (0, 1.75) {};
		\node [style=none] (14) at (0, 0.25) {};
		\node [style=none] (15) at (-0.4999999, 1) {$s$};
		\node [style=none] (16) at (7, 1.25) {$s$};
		\node [style=none] (17) at (7.5, 0.4999999) {};
		\node [style=none] (18) at (8.5, 1.25) {};
		\node [style=none] (19) at (7.5, 2) {};
		\node [style=none] (20) at (7.5, 1.25) {};
		\node [style=none] (21) at (6.75, -0.7500001) {$f(s)$};
		\node [style=none] (22) at (8.5, -0.7500001) {};
		\node [style=none] (23) at (7.5, -0) {};
		\node [style=none] (24) at (7.5, -0.7500001) {};
		\node [style=none] (25) at (-3, 1.75) {$\forall s$};
		\node [style=none] (26) at (6.75, -0) {};
		\node [style=none] (27) at (6.75, -1.5) {};
		\node [style=none] (28) at (7.5, -1.5) {};
	\end{pgfonlayer}
	\begin{pgfonlayer}{edgelayer}
		\draw (3.center) to (4.center);
		\draw (4.center) to (5.center);
		\draw (5.center) to (6.center);
		\draw (6.center) to (3.center);
		\draw (7.center) to (9.center);
		\draw (0.center) to (2.center);
		\draw (8.center) to (12.center);
		\draw [bend right=90, looseness=2.25] (13.center) to (14.center);
		\draw (13.center) to (14.center);
		\draw (20.center) to (18.center);
		\draw [bend right=90, looseness=2.25] (19.center) to (17.center);
		\draw (19.center) to (17.center);
		\draw (24.center) to (22.center);
		\draw (26.center) to (23.center);
		\draw (23.center) to (28.center);
		\draw (28.center) to (27.center);
		\draw [bend right=90, looseness=2.00] (26.center) to (27.center);
	\end{pgfonlayer}
\end{tikzpicture}

%% file: NDMProof2.tex
\begin{tikzpicture}[scale=0.5]
	\begin{pgfonlayer}{nodelayer}
		\node [style=none] (0) at (1, -0.5000001) {};
		\node [style=none] (1) at (1, 1) {};
		\node [style=none] (2) at (1, 1.75) {};
		\node [style=none] (3) at (1, -1.25) {};
		\node [style=none] (4) at (2.5, -1.25) {};
		\node [style=none] (5) at (2.5, 1.75) {};
		\node [style=none] (6) at (2.5, 1) {};
		\node [style=none] (7) at (2.5, -0.5000001) {};
		\node [style=none] (8) at (5, 0.9999998) {};
		\node [style=none] (9) at (1.75, 0.2500001) {$B$};
		\node [style=none] (10) at (3.5, -0.4999999) {};
		\node [style=none] (11) at (0, 0.9999998) {};
		\node [style=none] (12) at (3.5, 0.2500001) {};
		\node [style=none] (13) at (3.5, -1.25) {};
		\node [style=none] (14) at (3.999999, -0.4999999) {$e$};
		\node [style=none] (15) at (-7, 0.7500002) {};
		\node [style=none] (16) at (-6, 0.7500002) {};
		\node [style=none] (17) at (-6, 1.5) {};
		\node [style=none] (18) at (-6, -0) {};
		\node [style=none] (19) at (-4.5, -0) {};
		\node [style=none] (20) at (-4.5, 1.5) {};
		\node [style=none] (21) at (-4.5, 0.7500002) {};
		\node [style=none] (22) at (-3.5, 0.7500002) {};
		\node [style=none] (23) at (-1.75, 0.4999999) {$:=$};
		\node [style=none] (24) at (-5.25, 0.7500002) {$M_e$};
	\end{pgfonlayer}
	\begin{pgfonlayer}{edgelayer}
		\draw (2.center) to (3.center);
		\draw (3.center) to (4.center);
		\draw (4.center) to (5.center);
		\draw (5.center) to (2.center);
		\draw (6.center) to (8.center);
		\draw (7.center) to (10.center);
		\draw (11.center) to (1.center);
		\draw [bend left=90, looseness=2.25] (12.center) to (13.center);
		\draw (12.center) to (13.center);
		\draw (15.center) to (16.center);
		\draw (17.center) to (18.center);
		\draw (18.center) to (19.center);
		\draw (19.center) to (20.center);
		\draw (20.center) to (17.center);
		\draw (21.center) to (22.center);
	\end{pgfonlayer}
\end{tikzpicture}

%% file: NDMProof3.tex
\begin{tikzpicture}[scale=0.5]
	\begin{pgfonlayer}{nodelayer}
		\node [style=none] (0) at (2, -0.4999999) {};
		\node [style=none] (1) at (2, 0.9999998) {};
		\node [style=none] (2) at (2, 1.75) {};
		\node [style=none] (3) at (2, -1.25) {};
		\node [style=none] (4) at (3.5, -1.25) {};
		\node [style=none] (5) at (3.5, 1.75) {};
		\node [style=none] (6) at (3.5, 0.9999998) {};
		\node [style=none] (7) at (3.5, -0.4999999) {};
		\node [style=none] (8) at (6, 0.9999998) {};
		\node [style=none] (9) at (2.75, 0.2500001) {$B$};
		\node [style=none] (10) at (4.5, -0.4999999) {};
		\node [style=none] (11) at (0.9999998, 0.9999998) {};
		\node [style=none] (12) at (4.5, 0.2500001) {};
		\node [style=none] (13) at (4.5, -1.25) {};
		\node [style=none] (14) at (5, -0.4999999) {$e$};
		\node [style=none] (15) at (-7, 0.9999998) {};
		\node [style=none] (16) at (-6, 0.9999998) {};
		\node [style=none] (17) at (-6, 1.75) {};
		\node [style=none] (18) at (-6, 0.2500001) {};
		\node [style=none] (19) at (-4.5, 0.2500001) {};
		\node [style=none] (20) at (-4.5, 1.75) {};
		\node [style=none] (21) at (-4.5, 0.9999998) {};
		\node [style=none] (22) at (-3.5, 0.9999998) {};
		\node [style=none] (23) at (-1.75, 1) {$=$};
		\node [style=none] (24) at (-5.25, 0.9999998) {$M_e$};
		\node [style=none] (25) at (-7, 1.75) {};
		\node [style=none] (26) at (-7, 0.2500001) {};
		\node [style=none] (27) at (-7.5, 0.9999998) {$s$};
		\node [style=none] (28) at (0.9999998, 0.2500001) {};
		\node [style=none] (29) at (0.9999998, 0.9999998) {};
		\node [style=none] (30) at (0.9999998, 1.75) {};
		\node [style=none] (31) at (0.5000003, 0.9999998) {$s$};
		\node [style=none] (32) at (0.9999998, -3.5) {};
		\node [style=none] (33) at (2.25, -5.75) {};
		\node [style=none] (34) at (0.9999998, -4.25) {};
		\node [style=none] (35) at (2.75, -5.75) {$e$};
		\node [style=none] (36) at (0.5000003, -3.5) {$s$};
		\node [style=none] (37) at (0.9999998, -2.75) {};
		\node [style=none] (38) at (0.9999998, -3.5) {};
		\node [style=none] (39) at (3, -3.5) {};
		\node [style=none] (40) at (-1.75, -3.5) {$=$};
		\node [style=none] (41) at (2.25, -6.5) {};
		\node [style=none] (42) at (2.25, -5) {};
		\node [style=none] (43) at (1.5, -5) {};
		\node [style=none] (44) at (0.7500001, -5) {};
		\node [style=none] (45) at (1.5, -5.75) {};
		\node [style=none] (46) at (1.5, -6.5) {};
		\node [style=none] (47) at (0.7500001, -6.5) {};
		\node [style=none] (48) at (0.7500001, -5.75) {$f(s)$};
		\node [style=none] (49) at (-9.75, 2) {$\forall s$};
		\node [style=none] (50) at (0.9999998, -8) {};
		\node [style=none] (51) at (3, -8.75) {};
		\node [style=none] (52) at (0.9999998, -9.5) {};
		\node [style=none] (53) at (0.9999998, -8.75) {};
		\node [style=none] (54) at (-1.75, -8.75) {$\propto$};
		\node [style=none] (55) at (0.9999998, -8.75) {};
		\node [style=none] (56) at (0.5000003, -8.75) {$s$};
	\end{pgfonlayer}
	\begin{pgfonlayer}{edgelayer}
		\draw (2.center) to (3.center);
		\draw (3.center) to (4.center);
		\draw (4.center) to (5.center);
		\draw (5.center) to (2.center);
		\draw (6.center) to (8.center);
		\draw (7.center) to (10.center);
		\draw (11.center) to (1.center);
		\draw [bend left=90, looseness=2.25] (12.center) to (13.center);
		\draw (12.center) to (13.center);
		\draw (15.center) to (16.center);
		\draw (17.center) to (18.center);
		\draw (18.center) to (19.center);
		\draw (19.center) to (20.center);
		\draw (20.center) to (17.center);
		\draw (21.center) to (22.center);
		\draw [bend right=90, looseness=2.25] (25.center) to (26.center);
		\draw (25.center) to (26.center);
		\draw [bend right=90, looseness=2.25] (30.center) to (28.center);
		\draw (30.center) to (28.center);
		\draw [bend left=90, looseness=2.25] (42.center) to (41.center);
		\draw (42.center) to (41.center);
		\draw [bend right=90, looseness=2.25] (37.center) to (34.center);
		\draw (37.center) to (34.center);
		\draw (32.center) to (39.center);
		\draw (45.center) to (33.center);
		\draw (43.center) to (44.center);
		\draw (43.center) to (46.center);
		\draw (46.center) to (47.center);
		\draw [bend left=90, looseness=2.25] (47.center) to (44.center);
		\draw [bend right=90, looseness=2.25] (50.center) to (52.center);
		\draw (50.center) to (52.center);
		\draw (53.center) to (51.center);
	\end{pgfonlayer}
\end{tikzpicture}

%% file: LRIProof1.tex
\begin{tikzpicture}[scale=0.5]
	\begin{pgfonlayer}{nodelayer}
		\node [style=none] (0) at (1.5, -0) {};
		\node [style=none] (1) at (0.7499999, 0.7500001) {};
		\node [style=none] (2) at (0, 0.7500001) {};
		\node [style=none] (3) at (0.7499999, -0) {};
		\node [style=none] (4) at (0.7499999, -0.7500001) {};
		\node [style=none] (5) at (0, -0.7500001) {};
		\node [style=none] (6) at (0, -0) {$f(s)$};
		\node [style=none] (7) at (3, -0) {$=$};
		\node [style=none] (8) at (5.5, 0.75) {};
		\node [style=none] (9) at (5.5, -0.7500001) {};
		\node [style=none] (10) at (5.5, -0) {};
		\node [style=none] (11) at (6.25, -0) {};
		\node [style=none] (12) at (6.25, 0.75) {};
		\node [style=none] (13) at (7.75, 0.75) {};
		\node [style=none] (14) at (7.75, -0.7500001) {};
		\node [style=none] (15) at (7.75, -0) {};
		\node [style=none] (16) at (6.25, -0.7500001) {};
		\node [style=none] (17) at (8.5, -0) {};
		\node [style=none] (18) at (5, -0) {$a$};
		\node [style=none] (19) at (7, -0) {$X_s$};
	\end{pgfonlayer}
	\begin{pgfonlayer}{edgelayer}
		\draw (3.center) to (0.center);
		\draw (1.center) to (2.center);
		\draw (1.center) to (4.center);
		\draw (4.center) to (5.center);
		\draw [bend left=90, looseness=2.25] (5.center) to (2.center);
		\draw (10.center) to (11.center);
		\draw [bend right=90, looseness=2.25] (8.center) to (9.center);
		\draw (8.center) to (9.center);
		\draw (12.center) to (13.center);
		\draw (13.center) to (14.center);
		\draw (14.center) to (16.center);
		\draw (16.center) to (12.center);
		\draw (15.center) to (17.center);
	\end{pgfonlayer}
\end{tikzpicture}

%% file: LRIProof2.tex
\begin{tikzpicture}[scale=0.5]
	\begin{pgfonlayer}{nodelayer}
		\node [style=none] (0) at (1.5, -0) {};
		\node [style=none] (1) at (0.7499999, 0.7500001) {};
		\node [style=none] (2) at (0, 0.7500001) {};
		\node [style=none] (3) at (0.7499999, -0) {};
		\node [style=none] (4) at (0.7499999, -0.7500001) {};
		\node [style=none] (5) at (0, -0.7500001) {};
		\node [style=none] (6) at (0, -0) {$f(s)$};
		\node [style=none] (7) at (3, -0) {$=$};
		\node [style=none] (8) at (5.5, 0.75) {};
		\node [style=none] (9) at (5.5, -0.7500001) {};
		\node [style=none] (10) at (5.5, -0) {};
		\node [style=none] (11) at (6.25, -0) {};
		\node [style=none] (12) at (6.25, 0.75) {};
		\node [style=none] (13) at (7.75, 0.75) {};
		\node [style=none] (14) at (7.75, -0.7500001) {};
		\node [style=none] (15) at (7.75, -0) {};
		\node [style=none] (16) at (6.25, -0.7500001) {};
		\node [style=none] (17) at (8.5, -0) {};
		\node [style=none] (18) at (5, -0) {$b$};
		\node [style=none] (19) at (7, -0) {$Y_s$};
	\end{pgfonlayer}
	\begin{pgfonlayer}{edgelayer}
		\draw (3.center) to (0.center);
		\draw (1.center) to (2.center);
		\draw (1.center) to (4.center);
		\draw (4.center) to (5.center);
		\draw [bend left=90, looseness=2.25] (5.center) to (2.center);
		\draw (10.center) to (11.center);
		\draw [bend right=90, looseness=2.25] (8.center) to (9.center);
		\draw (8.center) to (9.center);
		\draw (12.center) to (13.center);
		\draw (13.center) to (14.center);
		\draw (14.center) to (16.center);
		\draw (16.center) to (12.center);
		\draw (15.center) to (17.center);
	\end{pgfonlayer}
\end{tikzpicture}

%% file: LRIProof3.tex
\begin{tikzpicture}[scale=0.5]
	\begin{pgfonlayer}{nodelayer}
		\node [style=none] (0) at (0, 2) {};
		\node [style=none] (1) at (0, -0) {};
		\node [style=none] (2) at (1, -0) {};
		\node [style=none] (3) at (1, 2) {};
		\node [style=none] (4) at (1, 2.75) {};
		\node [style=none] (5) at (1, -0.7500002) {};
		\node [style=none] (6) at (2.5, -0.7500002) {};
		\node [style=none] (7) at (2.5, 2.75) {};
		\node [style=none] (8) at (2.5, 2) {};
		\node [style=none] (9) at (2.5, -0) {};
		\node [style=none] (10) at (3.5, 2) {};
		\node [style=none] (11) at (3.5, -0) {};
		\node [style=none] (12) at (5, 1) {$=$};
		\node [style=none] (13) at (1.75, 1) {$T$};
		\node [style=none] (14) at (10, 2) {};
		\node [style=none] (15) at (7.5, -0) {};
		\node [style=none] (16) at (6.5, 2) {};
		\node [style=none] (17) at (7.5, 2.75) {};
		\node [style=none] (18) at (9, -0) {};
		\node [style=none] (19) at (7.5, 2) {};
		\node [style=none] (20) at (9, 2) {};
		\node [style=none] (21) at (9, 2.75) {};
		\node [style=none] (22) at (6.5, -0) {};
		\node [style=none] (23) at (10, -0) {};
		\node [style=none] (24) at (7.5, 1.25) {};
		\node [style=none] (25) at (9, 1.25) {};
		\node [style=none] (26) at (9, 0.7499999) {};
		\node [style=none] (27) at (7.5, 0.7499999) {};
		\node [style=none] (28) at (7.5, -0.7499999) {};
		\node [style=none] (29) at (9, -0.7499999) {};
		\node [style=none] (30) at (8.250001, 2) {$X$};
		\node [style=none] (31) at (8.250001, -0) {$Y$};
	\end{pgfonlayer}
	\begin{pgfonlayer}{edgelayer}
		\draw (4.center) to (5.center);
		\draw (5.center) to (6.center);
		\draw (6.center) to (7.center);
		\draw (7.center) to (4.center);
		\draw (8.center) to (10.center);
		\draw (9.center) to (11.center);
		\draw (1.center) to (2.center);
		\draw (0.center) to (3.center);
		\draw (21.center) to (17.center);
		\draw (20.center) to (14.center);
		\draw (18.center) to (23.center);
		\draw (22.center) to (15.center);
		\draw (16.center) to (19.center);
		\draw (17.center) to (24.center);
		\draw (24.center) to (25.center);
		\draw (25.center) to (21.center);
		\draw (27.center) to (26.center);
		\draw (26.center) to (29.center);
		\draw (29.center) to (28.center);
		\draw (28.center) to (27.center);
	\end{pgfonlayer}
\end{tikzpicture}